\documentclass[runningheads]{llncs}
\usepackage{graphicx}
\usepackage{amssymb,amsmath}
\usepackage{enumerate}
\usepackage{tikz}
\usepackage{caption}
\usepackage{algorithm}
\usepackage[noend]{algpseudocode}
\usepackage{lineno}

\newtheorem{observation}{Observation}

\DeclareMathOperator{\ov}{ov}
\DeclareMathOperator{\w}{w}

\DeclareMathOperator{\suff}{suff}
\DeclareMathOperator{\pref}{pref}
\newcommand{\s}{p}

\usetikzlibrary{calc,arrows,shapes,backgrounds,patterns,fit,decorations,decorations.pathmorphing}

\tikzstyle{hgedge}=[->,gray,line width=.4mm]
\tikzstyle{highhgedge}=[hgedge,yellow]
\tikzstyle{anypath}=[->,dashed]
\tikzstyle{vertex}=[draw,ellipse,inner sep=0.5mm]
\tikzstyle{inputvertex}=[draw,rectangle,inner sep=1mm]

\newenvironment{mypic}{\begin{center}\begin{tikzpicture}[line width=1.2pt,>=latex]}
		{\end{tikzpicture}\end{center}}

\begin{document}
	\title{Greedy Conjecture for the Shortest Common Superstring Problem and its Strengthenings}
	\titlerunning{Greedy Conjecture for SCS and its Strengthenings}
	%
	\author{Maksim S. Nikolaev\orcidID{0000-0003-4079-2885}}
	\authorrunning{M. S. Nikolaev}
	%
	\institute{Steklov Institute of Mathematics at St.~Petersburg, Russian Academy of Sciences
		\email{makc-nicko@yandex.ru}}
	\maketitle              
	\begin{abstract}
		In the~Shortest Common Superstring problem, one needs to~find the~shortest superstring for~a~set of~strings.
		This problem is APX-hard, and many approximation algorithms were proposed, with the~current best approximation factor of~2.466.
		Whereas these algorithms are technically involved, for more than thirty years the~Greedy Conjecture remains unsolved, that states that the~Greedy Algorithm ``take two strings with the~maximum overlap; merge them; repeat'' is a~2-approximation.
		
		This conjecture is still open, and one way to~approach it is to~consider its stronger version, which may make the~proof easier due to~the~stronger premise or provide insights from its refutation.
		In~this paper, we propose two directions to~strengthen the~conjecture.
		First, we introduce the Locally Greedy Algorithm (LGA), that selects a pair of strings not with the largest overlap but with the \emph{locally largest} overlap, that is, the largest among all pairs of strings with the same first or second string.
		Second, we change the~quality metric: instead of length, we evaluate the solution by~the~number of occurrences of~an~arbitrary symbol.
		
		Despite the double strengthening, we prove that~LGA is a~\emph{uniform} 4-approximation, that is, it always constructs a~superstring with no~more than four times as~many occurrences of an~arbitrary symbol as~any other superstring.		
		At the same time, we discover the limitations of the greedy heuristic: we show that~LGA is at least 3-approximation, and the Greedy Algorithm is at least uniform 2.5-approximation.
		These result show that if the Greedy Conjecture is true, it is not because the Greedy Algorithm is locally greedy or is uniformly 2-approximation.
		
		\keywords{superstring \and shortest common superstring \and approximation \and greedy algorithms \and greedy conjecture}
	\end{abstract}
	
	\section{Introduction}

In~the~Shortest Common Superstring problem (SCS) one is given a~set of~$n$ input strings $\mathcal{S} = \{ s_1, \dots, s_n \}$ and needs to~find the~shortest string that contains all of~them as~substrings.
SCS has applications in~genome assembly~\cite{waterman1995introduction,pevzner2001eulerian}: modern technologies can extract only relatively short sequences of~nucleotides, which leads to~the~problem of~reconstructing an~entire DNA from many such pieces. Without loss of generality, we can assume that no string of $\mathcal{S}$ is contained in any other.

SCS is NP-hard~\cite{GMS1980} and even APX-hard~\cite{BJLTY1991}.
The classic way to~get a~practical solution for~a~difficult optimization problem is to~use some kind of~heuristic.
Perhaps the~simplest heuristic for~SCS is~\emph{the greedy} one: to~find a~short superstring, merge strings with~the~longest \emph{overlap}, that~is, the~longest suffix of~one string that~is also a~prefix of~another.
An~algorithm that uses this heuristic is called \emph{the~Greedy Algorithm}~(GA).
It operates as~follows: while there are more than~two strings, choose a~pair with the~longest overlap and merge~it.
The~resulting string is clearly a~superstring for~the~input strings.

GA is not~deterministic, as~it may merge any pair with the~maximum overlap, and this can greatly affect the~length of the~result.
To~see this, consider a~dataset $\mathcal{S} = \{ \mathtt{ab}^n, \mathtt{b}^{n+1}, \mathtt{b}^n\mathtt{c} \}$, for which GA can construct both an~optimal solution $\mathtt{ab}^{n+1}\mathtt{c}$ and an~asymptotically twice as~long solution $\mathtt{ab}^n\mathtt{cb}^{n+1}$, depending on~the~tie-breaking rule.
This~example shows, that~the~approximation factor of GA is at~least two, and Tarhio and Ukkonen~\cite{TU1988} conjectured that it equals two for any tie-braking rule. This conjecture is known as~\emph{the~Greedy Conjecture}~(GC).

Over the past thirty-odd years, this conjecture has not been resolved.
This is strange, because usually the~analysis of such~simple greedy algorithms is also quite simple: either we understand that a~given algorithm can construct an~arbitrarily bad solution, or~we can obtain an~upper bound on~its~approximation factor and construct an~example for~which this bound turns~out to~be~tight.
The~first constant upper bound of~4 on~the~GA approximation factor was obtained by~Blum, Jiang, Li, Tromp, and Yannakakis~\cite{BJLTY1991}.
This~bound was improved to~3.5 by~Kaplan and Shafrir~\cite{KS2005}, and the~next two improvements: to~3.425 and to~3.396, respectively, were obtained recently by~Englert, Matsakis, and Vesel\'y~\cite{englert2022improved,englert2023simpler}.

Apart from that, the~following results related to~the~Greedy Algorithm are known:
\begin{itemize}
	\item GA is a factor $\frac12$-approximation for \emph{the maximum compression}~\cite{TU1988,T1989}, where compression is the difference between the total length of all input strings and the length of a~given superstring;
	\item GC holds for strings of length no more than 3~\cite{cazaux20143};
	\item GC holds for strings of length exactly 4~\cite{kulikov2015greedy};
	\item GC for the binary alphabet is equivalent to GC for the larger ones (see proof of Theorem 3 in~\cite{GMS1980});
	\item All instantiations of GA (that is, GA with specified tie-braking rule) achieve the same factor of approximation~\cite{nikolaev2021all}.
\end{itemize}
The Greedy Conjecture is challenging even for~rather simple special cases: for~instance, the~case of~strings of~length no~more than~4 is uncharted territory.

In~addition to~the~Greedy Algorithm, numerous approximation algorithms have been proposed (see Table~\ref{tb:approx}). The~current best upper bound of~2.466 on~the~approximation factor was recently obtained by Englert, Matsakis, and Vesel\'y~\cite{englert2023simpler}.
\begin{table}[h!]
	\centering
	\begin{tabular}{lll}
		$3.000$ & Blum, Jiang, Li, Tromp, Yannakakis~\cite{BJLTY1991} & 1991\\
		$2.889$ & Teng, Yao~\cite{TY1993} & 1993\\
		$2.834$ & Czumaj, Gasieniec, Piotrow, Rytter~\cite{CGPR1994} & 1994\\
		$2.794$ & Kosaraju, Park, Stein~\cite{KPS1994} & 1994\\
		$2.750$ & Armen, Stein~\cite{AS1994} & 1994\\
		$2.725$ & Armen, Stein~\cite{AS1995} & 1995\\
		$2.667$ & Armen, Stein~\cite{AS1996} & 1996\\
		$2.596$ & Breslauer, Jiang, Jiang~\cite{BJJ1997} & 1997\\
		$2.500$ & Sweedyk~\cite{S1999} & 1999\\
		$2.500$ & Kaplan, Lewenstein, Shafrir, Sviridenko~\cite{KLSS2005} & 2005\\
		$2.500$ & Paluch, Elbassioni, van~Zuylen~\cite{PEZ2012} & 2012\\
		$2.479$ & Mucha~\cite{M2013} & 2013\\
		$2.475$ & Englert, Matsakis, Vesel\'y~\cite{englert2022improved} & 2022\\
		$2.466$ & Englert, Matsakis, Vesel\'y~\cite{englert2023simpler} & 2023
	\end{tabular}
	\caption{Known upper bounds on approximation factors for SCS.}
	\label{tb:approx}
\end{table}

\subsection{Our Contribution}

We present \emph{the~Locally Greedy Algorithm}~(LGA), whose pseudocode is given in~Algorithm~\ref{algo:lga}.
\begin{algorithm}[h!]
	\caption{Locally Greedy Algorithm }
	\hspace*{\algorithmicindent} \textbf{Input:} set of strings~${\cal S}$.\\
	\hspace*{\algorithmicindent} \textbf{Output:} a~superstring~for $\mathcal{S}$.
	\begin{algorithmic}[1]
		\While{$\mathcal{S}$ contains at least two strings}
		\State extract from $\mathcal{S}$ any ordered pair $(s, t)$, whose overlap is the longest among pairs of the form $(s', t')$, where $s' = s$ or $t' = t$
		\State add to $\mathcal{S}$ the shortest string with a prefix $s$ and a suffix $t$
		\EndWhile
		\State return the only string from $\mathcal{S}$
	\end{algorithmic}
	\label{algo:lga}
\end{algorithm}

Clearly, GA is a~special case of LGA, and, as~GA, this algorithm makes kind~of~``greedy'' moves, but in~general they have little to~do with~the~length minimization.
As~an~example, consider a~set $\{ \mathtt{ab}^n, \mathtt{b}^n\mathtt{a} \}$. For~this~set, GA always constructs an~optimal superstring $\mathtt{ab}^n\mathtt{a}$, while LGA may well construct a~superstring $\mathtt{b}^n\mathtt{ab}^n$, which is asymptotically 2~times longer than the~optimal~one.
Although it seems implausible that~LGA is an~approximation algorithm at~all (and for~the~maximum compression it is indeed not: the example above shows that the total overlap of LGA can be arbitrarily less than the maximum), we modify the proof of~Blum et~al.~\cite{BJLTY1991} to~obtain
\[ |\mathrm{LGA}(\mathcal{S})| \leq 4\cdot|\mathrm{OPT}(\mathcal{S})|, \]
where $|s|$ is the length of a string $s$.

This result shows that to obtain a constant factor approximation for the length of the result, it is not necessary to focus on the reducing the length itself.
This means that LGA (and hence GA) can be a constant factor approximation not only for the shortest length, but also for other similar metrics that it does not directly minimize.
We show that it is the case for \emph{the number of occurrences of a given symbol}.
Let $|s|_p$ be the~number of~occurrences of a~symbol~$p$.
\begin{theorem}\label{theo:uni}
	For any symbol $p$ and a superstring $s$ for $\mathcal{S}$,
	\[ |\mathrm{LGA}(\mathcal{S})|_p \leq 4\cdot |s|_p. \]
\end{theorem}
We call this property a~\emph{uniform} factor 4~approximation.
Since $|s| = \sum_p |s|_p$, where the summation is taken over all the letters of the alphabet, the~uniform $\lambda$-approximation implies $\lambda$-approximation.

Theorem~\ref{theo:uni} may give the impression that the key to solving the Greedy Conjecture lies in the uniform approximation of GA, and we can try to prove it by showing that GA is a uniform 2-approximation. However, this is not the case.
\begin{theorem}\label{theo:lower}
	\emph{LGA} is at least 3-approximation. \emph{GA} is at least uniform 2.5-approximation.
\end{theorem}

As with $\lambda$-approximation, bounds on the uniform approximation factor for one instantiation of GA imply the same bounds for all instantiations:
\begin{theorem}\label{theo:equiv}
	If some instantiation of \emph{GA} is a uniform $\lambda$-approximation, then all of them are uniform $\lambda$-approximations.
\end{theorem}

The second result of Theorem~\ref{theo:lower} seems particularly interesting: we do not know if GA is a 2-approximation, but we do know for sure that it is not a uniform 2-approximation.
This means that if GA is indeed a 2-approximation, then it cannot be inefficient in terms of the number of occurrences for too many different symbols at once.
Why does this happen, and how can it be controlled?
On the other hand, this same result may mean that GA is not a 2-approximation, and the counterexample from Theorem~\ref{theo:lower} can serve as a starting point for constructing a counterexample to GC.
Finally, although a uniform approximation cannot be used to prove the 2-approximation of GA, it can potentially be used to prove the 2.5- or 3-approximation (if GA is indeed a 2.5 or 3-approximation, of course).

\subsection{Structure of the Paper}

In~Section~\ref{sec:defs} we introduce the~necessary notation and definitions.
In~Section~\ref{sec:uniform} we prove Theorem~\ref{theo:uni}.
In~Section~\ref{sec:lower} we prove Theorem~\ref{theo:lower}.
In~Section~\ref{sec:equiv} we prove Theorem~\ref{theo:equiv}.
Finally, in~Appendix~\ref{sec:proof} we prove the~general result, which implies Theorem~\ref{theo:uni} as a~special case.

\section{Preliminaries}\label{sec:defs}

For non-empty strings $s$ and $t$, by $\ov(s, t)$ we denote their \emph{overlap}, that is, the longest string $y$, such that $s = xy$ and $t = yz$ for some \emph{non-empty} strings $x$ and $z$, which we denote by $\pref(s, t)$ and $\suff(s, t)$, respectively.
A string $xyz = \pref(s, t) \ov(s, t) \suff(s, t)$ is~\emph{a merge} of~$s$ and~$t$ (see Fig.~\ref{fig:overlap}).
By $d(s, t) = |\pref(s, t)|$ we denote \emph{the distance} between $s$ and $t$.
By~$\varepsilon$ we denote the empty string.

\begin{figure}[ht]
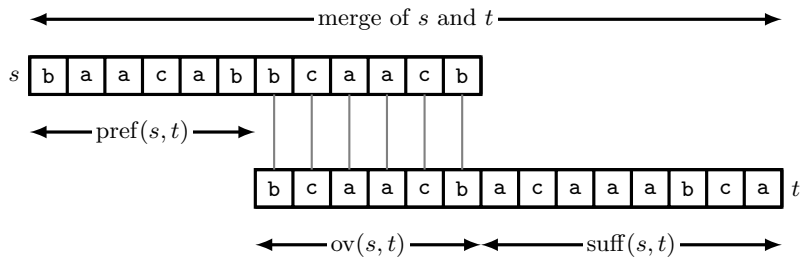

	\begin{mypic}
		\begin{scope}
			
			\draw (0,2) rectangle (6,2.5);
			\draw[step=5mm] (0,2) grid (6,2.5);
			\node[left] at (0,2.25) {$s$};
			\draw (3,0.5) rectangle (10,1);
			\draw[step=5mm] (3,0.5) grid (10,1);
			\node[right] at (10,0.75) {$t$};
			
			\foreach \x in {3.25, 3.75, ..., 5.75}
			\draw[gray,thick] (\x,2) -- (\x,1);
			
			\foreach \f/\t/\y/\lab in {0/3/1.5/{\pref(s,t)}, 
				3/6/0/{\ov(s,t)}, 6/10/0/{\suff(s,t)}, 0/10/3/{\text{merge of $s$ and $t$}}}
			\path (\f,\y) edge[<->] node[rectangle,inner sep=0.5mm,fill=white] {\strut $\lab$} (\t,\y);
			
			\foreach \x/\a in {0/b, 0.5/a, 1/a, 1.5/c, 2/a, 2.5/b, 3/b, 3.5/c, 4/a, 4.5/a, 5/c, 5.5/b}
			\node at (\x+0.25,2.25) {\tt \a};
			\foreach \x/\a in {3/b, 3.5/c, 4/a, 4.5/a, 5/c, 5.5/b, 6/a, 6.5/c, 7/a, 7.5/a, 8/a, 8.5/b, 9/c, 9.5/a}
			\node at (\x+0.25,0.75) {\tt \a};
		\end{scope}
	\end{mypic}
	\caption{Pictorial explanations of $\pref$, $\suff$, and $\ov$ functions.}
	\label{fig:overlap}
\end{figure}

Throughout the paper by $\mathcal{S} = \{ s_1, \dots, s_n \}$ we denote the set of $n$ input strings. We assume that no input string is a substring of another, since such strings may be found and removed efficiently. In this case, SCS becomes \emph{a permutation problem}: to find the shortest superstring, it is sufficient to find a permutation $(s_{\pi(1)}, \dots, s_{\pi(n)})$ that gives the shortest string after merging the adjacent strings. For a given permutation of indices $\pi$, the length of the corresponding superstring $s(\pi)$ is equal to
\begin{multline*}
	|s_{\pi(1)}| + |\suff(s_{\pi(1)}, s_{\pi(2)})| + \dots + |\suff(s_{\pi(n-1)}, s_{\pi(n)})| = \\
	= |\pref(s_{\pi(1)}, s_{\pi(2)})| + \dots + |\pref(s_{\pi(n-1)}, s_{\pi(n)})| + |s_{\pi(n)}| = \\
	= \sum_i |s_i| - \sum_{i = 1}^{n - 1} |\ov(s_{\pi(i)}, s_{\pi(i+1)})|.
\end{multline*}
The last expression shows that finding the shortest superstring is equivalent to maximizing \emph{the compression}, which is the sum of the overlaps of adjacent strings in a permutation. This naturally reduces SCS to the Longest Hamiltonian Path problem in \emph{an overlap graph} $OG(\mathcal{S})$, a complete directed graph $(V, E)$ with self-loops, where $V = \mathcal{S}$, $E = \mathcal{S} \times \mathcal{S}$ and the length of a directed edge $(s, t) \in E$ is $|\ov(s, t)|$ (see Figure~\ref{fig:og}).

The overlap graph has some crucial properties, which are widely used. One of them is \emph{the triangle inequality}:
\begin{gather}\label{eq:tri}
	d(u, w) \leq d(u, v) + d(v, w).
\end{gather}

Another one is \emph{Monge inequality}:
\begin{lemma}[Lemma 3.1 from~\cite{TU1988}]\label{lem:monge}
	If $|\ov(s, t)| \geq \max \{ |\ov(s, t')|, |\ov(s', t)| \}$, then $|\ov(s, t)| + |\ov(s', t')| \geq |\ov(s, t')| + |\ov(s', t)|$.
\end{lemma}

\begin{figure}
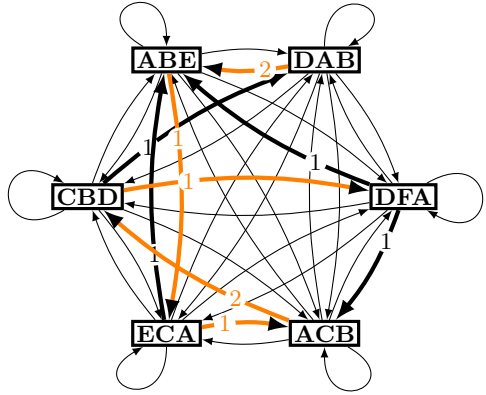

	\centering
	\begin{mypic}
		\begin{scope}[scale=0.7]
			\useasboundingbox (0,-1) rectangle (12,7);
			
			\begin{scope}[xshift=6cm,yshift=3cm]
				\foreach \n/\a in {DFA/0, DAB/60, ABE/120, CBD/180, ECA/240, ACB/300}
				\node [draw,rectangle,inner sep=.5mm] (\n) at (\a:30mm) {\bf \n};
			\end{scope}
			
			\foreach \n/\a in {DFA/0, DAB/60, ABE/120, CBD/180, ECA/240, ACB/300}
			\draw [->,line width=0.4pt] (\n) to [out=\a+30,in=\a-30,looseness=8] (\n);
			
			\foreach \s/\t/\w in {ACB/DAB/0,  ACB/ECA/0,  ACB/DFA/0,  ACB/ABE/0,  CBD/ACB/0,  CBD/ECA/0,  CBD/ABE/0,  DAB/ACB/0,  DAB/CBD/0,  DAB/ECA/0,  DAB/DFA/0, ECA/CBD/0,  ECA/DAB/0,  ECA/DFA/0,  DFA/CBD/0,  DFA/DAB/0,  DFA/ECA/0, ABE/ACB/0,  ABE/CBD/0,  ABE/DAB/0,  ABE/DFA/0
			}
			\path [->,draw=black,line width=0.4pt] (\s) edge[bend left=10] (\t);
			
			\foreach \s/\t/\w in {CBD/DAB/1, ECA/ABE/1,  DFA/ACB/1, DFA/ABE/1}
			\path [->,draw=black,line width=1.6pt] (\s) edge[bend left=10] node[near start,fill=white,circle,inner sep=0.1mm] {\small \textcolor{black}{\w}} (\t);
			
			\foreach \s/\t/\w in {DAB/ABE/2, ABE/ECA/1, ECA/ACB/1, ACB/CBD/2, CBD/DFA/1}
			\path [->,line width=1.6pt] (\s) edge[draw=orange,bend left=10] node[near start,fill=white,circle,inner sep=0.1mm] {\small \textcolor{orange}{\w}} (\t);
		\end{scope}
	\end{mypic}
	\caption{An overlap graph for $\mathcal{S} = \{ \mathtt{ABE}, \mathtt{DAB}, \mathtt{DFA}, \mathtt{ACB}, \mathtt{ECA}, \mathtt{CBD} \}$. Thin edges have zero weight. Orange path is a~Hamiltonian path of the~maximum weight 7, corresponding to the~shortest superstring \texttt{DABECACBDFA} of length $6\cdot3-7 = 11$.}
	\label{fig:og}
\end{figure}

For a cycle $C = v_1 \to \dots \to v_k \to v_1$ in the overlap graph, by $|C|$ we denote its \emph{length}:
\begin{gather}\label{eq:length}
	|C| := \sum_i (|v_i| - |\ov(v_i, v_{(i+1) \, \mathrm{mod}\, k})|) = \sum_i |\pref(v_i, v_{(i+1) \, \mathrm{mod}\, k})|.
\end{gather}

A set of edges is \emph{a cycle cover} if both the in-degree and the out-degree of any node is equal to 1.
This set is a collection of nonintersecting simple loops $C_1, \dots, C_k$.
\emph{The maximum weight cycle cover} in the overlap graph is a cycle cover $CC$ with the maximum total weight of its edges.
Note, that the maximum weight cycle cover is also \emph{the shortest} one, since $\sum_i |C_i| = \sum_i |s_i| - \sum_{(s, t) \in CC} |\ov(s, t)|$ and $\sum_i |s_i|$ is the same for all cycle covers.

Another property of the overlap graph is that nodes from the different cycles in the shortest cycle cover have bounded overlap.
\begin{lemma}[Lemma 9 from~\cite{BJLTY1991}]\label{lem:bound}
	Let $C_1$ and $C_2$ be two cycles in the shortest cycle cover with $c_1 \in C_1$ and $c_2 \in C_2$. Then
	\[ |\ov(c_1, c_2)| < |C_1| + |C_2|. \]
\end{lemma}

\section{Uniform 4-approximation of the Locally Greedy Algorithm}\label{sec:uniform}

\subsection{Pseudo-overlap Graph}

In this section, we introduce a general result that implies Theorem~\ref{theo:uni} as a special case.
Our approach is reminiscent to the one of Weinard and Schnitger~\cite{weinard2006greedy} and Laube and Weinard~\cite{laube2005conditional}: they considered some properties of the overlap graph (namely, the triangle inequality, Monge inequality and the triple inequality) and proved that they are insufficient to prove the Greedy Conjecture, by providing an example of a graph with all these properties, but on which the Greedy Algorithm constructs a superstring that is too long.
We also consider some properties of the overlap graph, but instead of proving insufficiency to achieve a factor 2 approximation, we prove sufficiency of these properties to achieve a factor 4 approximation.

Let $\mathcal{G} = (\mathcal{V}, \mathcal{V} \times \mathcal{V}, |\cdot|, \w)$, where $\mathcal{V} = \{ v_1, \dots, v_n \}$, be a complete directed graph with self-loops and non-negative weights on nodes and edges: weight of a node $v$ is $|v|$, weight of an edge $(u, v)$ is $\w(u, v)$. For a~subgraph $G$ of $\mathcal{G}$ (that is, $G$ is a pair $(V, E)$, where $V \subset \mathcal{V}$ and $E \subset V \times V$) define \emph{a weight} $\w(G) = \sum_{(u, v) \in E} \w(u, v)$ and \emph{a length} $\|G\| := \sum_{v \in V} |v| - \w(G)$. For example, the weight of the path consisting of one edge $u \to v$ is $\w(u, v)$ and its length is $|u| - \w(u, v) + |v|$. Note that $\|u \to v \to w\| \neq \| u \to v \| + \| v \to w\|$, since $|v|$ is added once on the left side and twice on the right.

As in the previous section, we may define \emph{the maximum weight cycle cover}, which is a cycle cover $C$ with maximum $\w(C)$. Note, that $C$ is also \emph{the shortest cycle cover} with respect to the length $\| \cdot \|$, since $\|C\| = \sum_{v \in \mathcal{V}} |v| - \w(C)$ and $\sum_{v \in \mathcal{V}} |v|$ is the same for all cycle covers. Since $\w \geq 0$, $C$ is also no longer than any Hamiltonian path in $\mathcal{G}$.

Let us call $\mathcal{G}$ \emph{a pseudo-overlap graph} if it satisfies the following four properties:
\begin{description}
	\item[\bf P1.]  $|u| \geq \w(u, v)$ and $|u| \geq \w(v, u)$ for all $u, v \in \mathcal{V}$.
	\item[\bf P2.] Triangle inequality: $|u| - \w(u, v) + |v| - \w(v, w) \geq |u| - \w(u, w)$ for all $u, v, w \in \mathcal{V}$.
	\item[\bf P3.] Monge inequality: if $\w(u, v) \geq \max\{ \w(u, v'), \w(u', v) \}$, then
	\[ \w(u, v) + \w(u', v') \geq \w(u, v') + \w(u', w). \]
	\item[\bf P4.] If $C_1$ and $C_2$ are two cycles in the maximum weight cycle cover on $\mathcal{G}$ with $c_1 \in C_1$ and $c_2 \in C_2$, then $\w(c_1, c_2) \leq \| C_1 \| + \| C_2 \|$.
\end{description}
Note that \textbf{P1} ensures that the length of any path or cycle is nonnegative, and \textbf{P2} ensures that a path $u \to v \to w$ is at least as long as $u \to w$.
This means, in particular, that the shortest cycle cover on a complete subgraph $G = (V, V\times V)$ of $\mathcal{G}$ is no longer than the shortest cycle cover on $\mathcal{G}$.

Let us say that an edge $(u, v)$ \emph{dominates} (\emph{strictly dominates}) another edge $(u', v')$, if $u' = u$ or $v' = v$ and $\w(u, v) \geq \w(u', v')$ ($\w(u, v) > \w(u', v')$).  Call $(u, v)$ \emph{dominant}, if it dominates every edge, that begins in $u$ or ends in $v$. Recall that $G$ includes self-loops, so $u$ may well be $v$, $u'$, $v'$, and so on.

Consider an algorithm PATH that goes through a list of all edges in $G$ in \emph{a dominance respecting order}, that is, if $(u, v)$ strictly dominates $(u', v')$, then $(u, v)$ will be processed earlier, and includes some of them in the final solution. Specifically, PATH \emph{does not} include another edge $(u, v)$ if and only if
\begin{description}
	\item[\bf R1.] it is dominated by an already chosen edge,
	\item[\bf R2.] it is not dominated but it would form a cycle.
\end{description}

The resulting set of chosen edges forms a Hamiltonian path $\mathrm{PATH}(\mathcal{G})$.

\begin{theorem}\label{theo:path}
	Let $SHP$ be the shortest Hamiltonian path in $\mathcal{G}$. Then $\|\mathrm{PATH}(\mathcal{G})\| \leq 4\cdot\|SHP\|$.
\end{theorem}

The proof is in Appendix~\ref{sec:proof}.

\begin{corollary}
	The Locally Greedy Algorithm is a factor 4 approximation.
\end{corollary}
\begin{proof}
	Note that LGA is a special case of PATH and the overlap graph is indeed a pseudo-overlap graph, where $|s|$ is the length of the corresponding string, $\w(s, t) = |\ov(s, t)|$ and the length of the path is the length of the corresponding superstring: \textbf{P1} and \textbf{P2} are obvious, \textbf{P3} is valid by Lemma~\ref{lem:monge} and \textbf{P4} is valid by Lemma~\ref{lem:bound}.
\end{proof}

\subsection{Number of Occurrences Instead of Lengths}

In this section, we prove that if we replace the length of a string by the number of occurrences of an arbitrary symbol ${\s}$ in it, the corresponding overlap graph remains a pseudo-overlap graph, which implies Theorem~\ref{theo:uni}. More specifically, for any non-empty string $\s$, let $\mathcal{G}_{\s}$ be a graph such that $\mathcal{V} = \mathcal{S}$,  $|s| = |s|_{\s}$ and $\w(s, t) = |\ov(s, t)|_{\s}$.

\begin{theorem}\label{theo:pseudo}
	$\mathcal{G}_{\s}$ is a pseudo-overlap graph.
\end{theorem}
\begin{proof}
	\begin{description}
		\item[\textbf{P1.}] This property clearly holds.
		
		\item[\textbf{P2.}] Since $\pref(u, v)\pref(v, w)$ contains $\pref(u, w)$, the number of occurrences in it is at least the number of occurrences in $\pref(u, w)$.
		
		\item[\textbf{P3.}] If $|\ov(u, v)|_{\s} \geq |\ov(u, v')|_{\s} + |\ov(u', v)|_{\s}$ then this property holds. If $|\ov(u, v)|_{\s} < |\ov(u, v')|_{\s} + |\ov(u', v)|_{\s}$ then there are strings $x, y, z$ such that $\ov(u, v) = xyz$, $\ov(u, v') = yz$, $\ov(u', v) = xy$. Since $y$ is a suffix of $u'$ and also a prefix of $v'$, then $\ov(u', v')$ contains $y$ and hence $|\ov(u', v')|_{\s} \geq |y|_{\s}$. Finally,
		\[ |\ov(u, v)|_{\s} = |xyz|_{\s} = |xy|_{\s} + |yz|_{\s} - |y|_{\s} \geq |\ov(u, v')|_{\s} + |\ov(u', v)|_{\s} - |\ov(u', v')|_{\s}. \]
		
		\item[\textbf{P4.}]
		
		For a string $s$ by $s[i]$, we denote the $i$-th symbol of $s$. If $s[i] = \s$, then $i$ is \emph{an occurrence} of ${\s}$ in $s$.
		Let $s(\pi)$ be a superstring, corresponding to the permutation $(s_{\pi(1)}, \dots, s_{\pi(n)})$.  Then 
		\begin{gather}
			|s(\pi)|_{\s} = \sum_{i=1}^{n-1}|\pref(s_{\pi(i)}, s_{\pi(i+1)})|_\s + |s_{\pi(n)}|_{\s}
		\end{gather}
		
		Similarly, if $C = v_1 \to \dots \to v_k \to v_1$ is a cycle in the overlap graph, by $|C|_{\s}$ denote the following:
		\begin{gather}\label{eq:Cp}
			|C|_{\s} := \sum_{i = 1}^k |\pref(v_i, v_{(i+1) \, \mathrm{mod} \, k})|_\s.
		\end{gather}
		
		\begin{lemma}\label{lem:tech}
			Let $C = v_1 \to \dots \to v_k \to v_1$ be a cycle in the overlap graph and $v \in C$. If $l \leq r \leq |v|$ and $r - l < |C|$, then the number of occurrences of ${\s}$ in the substring $v[l, r]$ of $v$ from $l$-th to $r$-th symbol is not greater than $| C |_{\s}$.
		\end{lemma}
		\begin{proof}
			Without loss of generality, assume that $v = v_1$. To begin with, note, that if $s = w^ku$, $k > 0$, is a prefix of an infinite cyclic string $w^\infty$ with the period $w$, then $|s[i, i+|w|-1]|_\s$ is the same for all $i$ and is equal to $|w|_\s$. Also, if $[l', r'] \subset [l, r]$, then $|s[l', r']|_\s \leq |s[l, r]|_\s$. Therefore, if $r - l < |w|$, then
			\begin{gather}\label{eq:1wlr}
				|w|_\s = |s[1, |w|]|_\s \geq |s[l, r]|_\s. 
			\end{gather}
			
			Let $w = \pref(v_1, v_2)\pref(v_2, v_3)\dots\pref(v_k, v_1)$.
			Then $v = v_1$ is a prefix of $w^\infty$ and $|w| = |C|$, $|w|_\s = |C|_\s$.
			Therefore, if $r - l < |C|$, then by~\eqref{eq:1wlr} $|v[l, r]|_\s \leq |C|_{\s}$.
		\end{proof}
		
		Now we are ready to prove that property \textbf{P4} holds.
		\begin{lemma}\label{lem:ubound}
			Let $C_1$ and $C_2$ be two cycles in the shortest cycle cover with $c_1 \in C_1$ and $c_2 \in C_2$. Then
			\[ |\ov(c_1, c_2)|_{\s} < |C_1|_{\s} + |C_2|_{\s}. \]
		\end{lemma}
		\begin{proof}
			By~Lemma~\ref{lem:bound}, $|\ov(c_1, c_2)| \leq |C_1| + |C_2|$.
			If $|\ov(c_1, c_2)| \leq |C_i|$ for some $i \in \{1, 2\}$, then $\ov(c_1, c_2) = c_i[l, r]$ for some $l \leq r$, $r - l < |C_i|$ and Lemma~\ref{lem:tech} finishes the proof.
			
			Otherwise, let $\ov(v, v') = xy$, where $|x| = |C_1|$ and $|y| \leq |C_2|$.
			Then $x = c_1[l_x, r_x]$, where $r_x - l_x = |C_1| - 1$, and $y = c_2[l_y, r_y]$, where $r_y - l_y < |C_2|$.
			Again, by Lemma~\ref{lem:tech}, $|\ov(c_1, c_2)|_{\s} = |x|_\s + |y|_\s \leq \|C_1\|_{\s} + \|C_2\|_{\s}.$
		\end{proof}
	\end{description}
\end{proof}

\subsection{Proof of Theorem~\ref{theo:uni}}

Note that if $|\ov(s, t)| \geq |\ov(s, t')|$ then $|\ov(s, t)|_\s \geq |\ov(s, t')|_\s$ for any $\s$. Similarly, if $|\ov(s, t)| \geq |\ov(s', t)|$ then $|\ov(s, t)|_\s \geq |\ov(s', t)|_\s$. This means that LGA is an instantiation of PATH for all $\mathcal{G}_\s$ simultaneously.

Theorem~\ref{theo:pseudo} and the fact, that $|\mathrm{PATH}(\mathcal{G}_\s)|_\s$ is precisely $\| \mathrm{PATH}(\mathcal{G}_\s) \|$, finish the proof.

\section{Lower bounds}\label{sec:lower}

\subsection{LGA is at Least 3-approximation}

The main source of suboptimality of LGA is that if there are two strings $s$ and $t$ such that $s$ has only empty overlaps to the right and $t$ has only empty overlaps to the left, it can merge them, even if there are much better options. In particular, if there are only two remaining strings, LGA can merge them in any order.

Consider a set $\mathcal{S} = \{ \mathtt{a}\mathtt{b}^n, \mathtt{b}^{n+1}, \mathtt{b}^n \mathtt{c}, \mathtt{b}^{n-1}\mathtt{c}^2 \}$.
The shortest solution $\mathtt{a} \mathtt{b}^{n+1} \mathtt{c}^2$ has length $n+4$.
Let us consider the following solution:
\[ \mathtt{b}^{n-1}\mathtt{c}^2 \xrightarrow{2} \mathtt{a}\mathtt{b}^n \xrightarrow{1} \mathtt{b}^n \mathtt{c} \xrightarrow{3} \mathtt{b}^{n+1}, \]
where superscripts indicate the order of merges.
Here the first merge is valid due to the largest overlap, the second is valid because $\mathtt{b}^{n-1}\mathtt{c}^2$ has empty overlap to the right with both remaining strings, and $\mathtt{a}\mathtt{b}^n$ has empty overlap to the left with both remaining strings, and the last merge is always valid thanks to the observation from the beginning of the section.
The corresponding superstring $\mathtt{b}^{n-1}\mathtt{c}^2\mathtt{a}\mathtt{b}^n\mathtt{c}\mathtt{b}^{n+1}$ has the length $3n+4$, which is asymptotically 3 times longer than $n+4$.

This example is similar to the dataset $\{ \mathtt{a}\mathtt{b}^n, \mathtt{b}^{n+1}, \mathtt{b}^n \mathtt{c} \}$ from the introduction, where GA could get twice as long a solution, but now there are two strings that we may place one to the left and one to the right of $\mathtt{a}\mathtt{b}^n\mathtt{c}$, so the resulting solution is three times longer.

It is interesting that if we continue and consider, say, a dataset
\[ \{ \mathtt{a}\mathtt{b}^n, \mathtt{b}^{n+1}, \mathtt{b}^n \mathtt{c}, \mathtt{b}^{n-1}\mathtt{c}^2, \mathtt{b}^{n-2}\mathtt{c}^3 \}, \]
then we will not be able to get a 4 times longer solution, because after we get the string $\mathtt{b}^{n-2}\mathtt{c}^3 \mathtt{a}\mathtt{b}^n\mathtt{c}$, it and $\mathtt{b}^{n-1}\mathtt{c}^2$ will have long overlaps on both sides.

\subsection{GA is at Least Uniform 2.5-approximation}\label{sec:2.5}

Consider a set $\mathcal{S} = \{ \mathtt{aaaab}, \mathtt{aaabaa}, \mathtt{aabaaba}, \mathtt{baabaa}, \mathtt{abaaaa} \}$. A solution \[ \mathtt{aaaabaabaaaa} \] has 2 occurrences of $\s = \mathtt{b}$.

Let us consider the following solution:
\[ \mathtt{baabaa} \xrightarrow{1} \mathtt{aabaaba} \xrightarrow{4} \mathtt{aaabaa} \xrightarrow{2} \mathtt{abaaaa} \xrightarrow{3} \mathtt{aaaab}. \]
The first merge is valid due to the longest possible overlap of length 5, the second and third are valid due to the longest possible overlaps of length 4, and the last is valid due to the longest possible overlap of length 1.
The corresponding superstring $\mathtt{baabaabaaabaaaab}$ has~5 occurrences of $\mathtt{b}$.

This example was found through computer search and creates more questions than answers.
First, it is not clear whether this example can be tweaked to achieve uniform 3-approximation of GA.

Second, it seems like this example cannot be easily modified into a set for which GA would construct a solution more than twice as long as the shortest one. A natural approach in this direction would be to replace the symbol $\mathtt{b}$ with $\mathtt{b}^k$, but in this case this would also require increasing the number of $\mathtt{a}$'s to keep the chosen overlaps maximal, and this does not result in more than two times longer solution.

\section{All Instantiations of GA Achieve the Same Factor of the Uniform Approximation}\label{sec:equiv}

\begin{proof}[Proof of Theorem~\ref{theo:equiv}]
	
	The proof of this result is based on the technique from~\cite{nikolaev2021all}, which was used to prove that if there is a factor $\lambda$ approximation instanitation $A$ of GA, then all instantiations are $\lambda$-approximations.
	The main idea is the following: suppose to the contrary, that there is an instantiation $B$ of GA and a set $\mathcal{S}$, such that $|B(\mathcal{S})| > \lambda |\mathrm{OPT}(\mathcal{S})|$.
	Using this set and the order of the merges of $A$, construct a new set $\mathcal{S}'$, such that $|B(\mathcal{S}')| > \lambda |\mathrm{OPT}(\mathcal{S}')|$, but there is only one greedy order of \emph{non-trivial} merges (that is, greedy merges with non-empty overlap), so $|A(\mathcal{S}')| = |B(\mathcal{S}')|$.
	Then $A$ is not $\lambda$-approximation, which contradicts its definition.
	
	Let $A$ be a uniform $\lambda$-approximation instantiation of GA, and suppose to the contrary that there are another instantiation $B$, a set $\mathcal{S}$ and a symbol $\s$, such that $| B(\mathcal{S}) |_{\s} > \lambda | s(\pi_{OPT}) |_{\s}$, where $\pi_{OPT}$ is some permutation. Let $B(\mathcal{S}) = t_1 t_2 \dots t_k$, where $\{ t_1, \dots, t_k \}$ is a set of strings, constructed by $B$ up to the moment before the first merge with an empty overlap. Clearly, $|B(\mathcal{S})|_{\s} = \sum_i |t_i|_{\s}$
	
	For every string $s_i = c_1c_2\dots c_{|s_i|} \in \mathcal{S}$ define a string
	\begin{gather*}\label{eq:dist}
		s'_i = \$^{m-\alpha_i} c_1 \$^m c_2 \$^m c_3 \$^m \dots \$^m c_{|s_i|} \$^{\beta_i},
	\end{gather*}
	where $\$$ is \emph{a sentinel}, a symbol that does not occur in $\mathcal{S}$, and $m$, $\alpha_i < m$ and $\beta_i < m$ are chosen in such a way, that all the instantiations of GA construct the same set $\{ t'_1, \dots, t'_k \}$ before the first trivial merge, where $t'_i$ is obtained by merging strings from $\mathcal{S}'$ with the same indices and in the same order as in $t_i$.
	The existence of such $m$, $\alpha_i$ and $\beta_i$ is proved in Section~3 of~\cite{nikolaev2021all}.
	
	Note, that $|s'_i|_{\s} = |s_i|_{\s}$ and $|\ov(s'_i, s'_j)|_{\s} = |\ov(s_i, s_j)|_{\s}$.
	Hence, $|s'(\pi)|_{\s} = |s(\pi)|_{\s}$ for each $\pi$ and $|t'_i|_{\s} = |t_i|_{\s}$ for each $i$. Therefore,
	\[ | A(\mathcal{S}') |_{\s} = \sum_i |t'_i|_{\s} = \sum_i |t_i|_{\s} = |B(\mathcal{S})|_{\s} > \lambda |s(\pi_{OPT})|_{\s} = \lambda |s'(\pi_{OPT})|_{\s}, \]
	and $A$ is not uniform $\lambda$-approximation. This contradiction with the definition of $A$ finishes the proof.
\end{proof}

\section{Proof of Theorem~\ref{theo:path}}\label{sec:proof}

As was mentioned before, this proof is a modification of the proof of the 4-approximation of the Greedy Algorithm, presented in~\cite{BJLTY1991}. The overall structure is similar, and the main goal is to get rid of all arguments that appeal to the fact that the Greedy Algorithm merges pairs of strings with the largest overlap.
For many statements below, their counterparts from~\cite{BJLTY1991} are indicated in brackets after the title.

Consider an algorithm CYC that goes through a list of all edges in $\mathcal{G}$ in the dominance respecting order and does not include an edge $(u, v)$ in the final solution if and only if it is dominated by an already chosen edge. The resulting set of edges is a cycle cover on $\mathcal{G}$, and Monge inequality ensures that it is the shortest cycle cover.

\begin{lemma}[Theorem~10]\label{lem:cyc}
	Let $C$ be a cycle cover constructed by \emph{CYC}. Then $C$ is a shortest cycle cover.
\end{lemma} 
\begin{proof}
	Among the shortest cycle covers consider a cycle cover $C'$ that has the maximum number of edges in common with $C$. We prove that $C'$ is, in fact, equal to $C$.
	
	Let $(u, v')$ be an edge in $C' \setminus C$ with the maximum weight. Since $(u, v')$ was not chosen by CYC, it is dominated by some edge in $C$. Without loss of generality, assume that this edge is of form $(u, v)$. As $(u, v) \notin C'$, there must be an edge of form $(u', v)$. Since $\w(u, v) \geq \w(u, v') \geq \w(u', v)$, where second inequality holds by the definition of $(u, v')$, the edge $(u, v)$ is also dominating $(u', v)$. By Monge inequality, replacing $(u, v')$ and $(u', v)$ with $(u, v)$ and $(u', v')$ results in a cycle cover $C''$ with at least as large weight (and hence at most as large length) as $C'$, that has at least one more edge in common with $C$. This contradicts the definition of $C'$.
\end{proof}

Now let us return to the algorithm PATH. The resulting set of chosen edges $\mathrm{PATH}(\mathcal{G})$ forms a Hamiltonian path $v_{\pi(1)} \to \dots \to v_{\pi(n)}$ for some permutation~$\pi$. For convenience, let us renumber $V$ so that $v_{\pi(i)}$ becomes $v_i$.

Call an edge $(v_j, v_i)$, rejected because of \textbf{R2}, \emph{a bad back edge}, which is ``back'' because necessarily $i \leq j$.
\begin{observation}\label{obs:earlier}
	If $(v_j, v_i)$ is a bad back edge, then all edges $(v_i, v_{i+1}), \dots, (v_{j-1}, v_{j})$ were processed earlier and edges $(v_{i-1}, v_i)$ and $(v_{j}, v_{j+1})$ will be processed later.
\end{observation}
Let us say that a bad back edge $(v_j, v_i)$ \emph{spans an interval} $[i, j]$. Observation~\ref{obs:earlier} implies that
\begin{observation}[Lemma 13]
	The intervals $[i, j]$ and $[i', j']$ spanned by two bad back edges are disjoint, or one contains the other.
\end{observation}
Thus, the bad back edges do not cross each other.
Call a path $v_i \to \dots \to v_j$, covered by some innermost bad back edge $(v_j, v_i)$, \emph{a culprit}.
Call the edge included last between two successive culprits \emph{a weak link}
\footnote{
	In~\cite{BJLTY1991}, a weak link is \emph{the shortest} edge between two successive culprits.
	Of course, the shortest edge is the last included in the case of GA.
	Authors remarked that weak links are also shorter than all the edges in the corresponding culprits, but they did not use this property.
}.
If the weak links are removed, the Hamiltonian path $s_1 \to \dots \to s_n$ falls into \emph{blocks}, each of which consists of the culprit as the middle segment and, possibly empty, \emph{left} and \emph{right extensions}. Thus $\mathcal{V}$ is divided into sets $\mathcal{V}_l$, $\mathcal{V}_m$ and $\mathcal{V}_r$ of the left, middle, and right nodes, respectively.
\begin{observation}[Lemma 14]\label{obs:14}
	Let $(v_j, v_i)$ be a bad back edge. Node $v_i$ is the left node or the first node of a culprit. Node $v_j$ is either a right node or the last node of a culprit.
\end{observation}
\begin{proof}
	If $v_i$ is a right node, then an edge $(v_{i-1}, v_i)$ lies on the left of the corresponding weak link.
	By observation~\ref{obs:earlier} this weak link was included before $(v_j, v_i)$ was processed, and hence before $(v_{i-1}, v_i)$ was included. This contradicts the definition of a weak link. A similar argument holds for $v_j$. 
\end{proof}

Consider the graph $G_l = (V_l, E_l)$, where $V_l = \mathcal{V}_l \cup \mathcal{V}_m$ and $E_l$ is a set of all chosen non-weak edges between left and middle nodes. Let $C_l$ be the shortest cycle cover on $(V_l, V_l \times V_l)$. Let $V_r = \mathcal{V}_m \cup \mathcal{V}_r$ and define similarly the graph $G_r$ and the shortest cycle cover $C_r$.

Let $V$ be the disjoint union $V_l \sqcup V_r$, $E := E_l \sqcup E_m$ and $G := G_l \sqcup G_m$. Thus, each left or right node occurs in $V$ once, while each middle node occurs twice. Add each weak link to $E$ as an edge from the last node of the corresponding culprit/right extension in $G_r$ to the first node of the corresponding culprit/left extension in $G_l$. Denote the resulting set of edges by $E'$. Note, that
\begin{gather}\label{eq:wE}
	\w(E') = \w(E_l) + \w(E_r) + \w(WL),
\end{gather}
where $WL$ is a set of weak links.

Consider $C = C_l \sqcup C_r$, a cycle cover on $(V, V \times V)$. Its weight equals
\begin{gather}\label{eq:wC}
	\w(C) = \w(C_l) + \w(C_r).
\end{gather}
Each edge of $C$ connects two $V_l$ nodes or two $V_r$ nodes, so all edges of $C$ satisfy the assumptions of the following lemma.
\begin{lemma}[Lemma 16]\label{lem:dom}
	Let $C'$ be any cycle cover on $V$. Let $(u, v)$ be an edge of $C' \setminus E'$ not in $V_r \times V_l$. Then $(u, v)$ is dominated by either
	\begin{enumerate}
		\item an adjacent $E'$ edge,
		\item a bad back edge of the culprit with which it shares the head $v$ and $v \in V_r$, or
		\item a bad back edge of the culprit with which it shares the tail $u$ and $u \in V_l$.
	\end{enumerate}
\end{lemma}
\begin{proof}
	Suppose first that $(u, v)$ corresponds to a bad back edge.
	By Observation~\ref{obs:14}, $v$ corresponds to a left node or to the first node of a culprit.
	In the latter case, $(u, v)$ is dominated by the back edge of the culprit: since it is inside $(u, v)$, it was processed before in the dominance respecting order.
	Therefore, either $v$ is the first node of a culprit (and case (2) holds), or else $v \in V_l$.
	Similarly, either $u$ is the last node of a culprit (and case (3) holds), or else $u \in V_r$.
	Since $(u, v)$ is not in $V_r \times V_l$, it follows that case (2) or case (3) holds.
	
	Suppose that $(u, v)$ does not correspond to a bad back edge.
	Then it must be dominated by some edge chosen by PATH.
	If this edge is in $E'$ then case (1) holds.
	If it is not in $E'$, then $(u, v)$ shares the head or tail with the bad back edge of some culprit, so (2) or (3) holds.
\end{proof}

Although Lemma~\ref{lem:dom} ensures that each edge of $C$ is dominated either by edge of $E'$ or by bad back edge of culprit, it may be that some edges of $E'$ dominate both of their adjacent edges of $C$. The following lemma shows that we can modify $C$ into a new cycle cover $C'$ with at least the same weight, so that each edge of $E'$ dominates no more than one of its adjacent edges of $C'$.
\begin{lemma}[Lemma 17]\label{lem:dom2}
	Let $C$ be any cycle cover on $V$ such that $C \setminus E'$ does not contain edges of $V_r \times V_l$. Then there is a cycle cover $C'$ such that
	\begin{enumerate}
		\item $C' \setminus E'$ has also no edges from $V_r \times V_l$,
		\item $\w(C') \geq \w(C)$,
		\item each edge in $E' \setminus C'$ dominates no more than one of its two adjacent $C'$ edges.
	\end{enumerate}
\end{lemma}
\begin{proof}
	Since $C$ already has the first two properties, it is sufficient to argue that if $C$ violates property (3), then we can construct another cycle cover $C'$ that satisfies properties (1) and (2), and has more edges in common with $E'$.
	
	Let $(u, v)$ be an edge from $E' \setminus C$ that dominates both adjacent $C$ edges $(u', v)$ and $(u, v')$. By Monge inequality, replacing edges $(u', v)$ and $(u, v')$ with $(u, v)$ and $(u', v')$ produces a cover $C'$ with at least as much weight. To see that the new edge $(u', v')$ is not in $V_r \times V_l$, observe that if $u' \in V_r$ then $v \in V_r$ ($C \setminus E'$ has no edges in $V_r \times V_l$), which implies that $u \in V_r$ ($E'$ has no edges in $V_l \times V_r$ by construction), which implies that also $v' \in V_r$ because of $(u, v') \in C \setminus E'$.
\end{proof}

By Lemmas~\ref{lem:dom} and~\ref{lem:dom2}, we can construct from the cycle cover $C$ another cycle cover $C'$ with at least as large weight, and such that each edge of $C'$ is dominated by the edge of $E'$ or by the bad back edge of the culprit. The edges of $E'$ do not dominate more than one edge of $C'$ and the bad back edges of the culprits do not dominate more than two. Thus, by~\eqref{eq:wC},
\begin{gather}\label{eq:C'}
	\w(C_l) + \w(C_r) = \w(C) \leq \w(C') \leq \w(E') + 2\w(BC),
\end{gather}
where $BC$ is a set of bad back edges of culprits.

Let $C_m$ be a cycle cover on $\mathcal{V}_m$, where each cycle is a culprit path $P_i$ closed by the corresponding bad back edge.
\begin{lemma}[Lemma 15]\label{lem:maxweight}
	$C_m$ is the shortest cycle cover on $\mathcal{V}_m$.
\end{lemma}
\begin{proof}
	Recall the algorithm CYC that goes through a list of all edges in $\mathcal{V}_m \times \mathcal{V}_m$ in a dominance respecting order, induced by the order of PATH in $\mathcal{V}$, and does not include another edge if and only if it is dominated by an already chosen edge. Clearly, it constructs $C_m$.	
	Therefore, by Lemma~\ref{lem:cyc}, $C_m$ is the shortest cycle cover.
\end{proof}

Note that $\|\mathrm{PATH}(\mathcal{G})\| = \| E' \| - \sum_i \| P_i \|$ and $\sum_i \| P_i \| = \| C_m \| + \w(BC)$. Then 
\begin{align}
	\nonumber
	\|\mathrm{PATH}(\mathcal{G})\| & = \| E' \| - \sum_i \| P_i \| \\
	\nonumber
	& = \sum_{v \in V_r} |v| + \sum_{v \in V_l} |v| - \w(E') - \sum_i \| P_i \| \\
	\nonumber
	& \overset{\eqref{eq:C'}}{\leq} \sum_{v \in V_r} |v| + \sum_{v \in V_l} |v| - \w(C_l) - \w(C_r) + 2 \w(BC) - \sum_i \| P_i \| \\
	\nonumber
	& = \| C_l \| + \| C_r \| + 2 \w(BC) - \sum_i \| P_i \| \\
	\label{eq:main1}
	& \leq 2\cdot\|SHP\| + \w(BC) - \| C_m \|.
\end{align}

\begin{lemma}[Theorem 8]
	\begin{gather}\label{eq:main2}
		\w(BC) - 2\|C_m\| \leq \|SHP\|.
	\end{gather}
\end{lemma}
\begin{proof}
	Suppose $C_m = C_1 \cup \dots \cup C_d$ and let $c_i \in C_i$ be a node with the maximum weight $|c_i|$.
	By \textbf{P1}, the weight of the bad back edge of $C_i$ is not greater than $|c_i|$.
	By \textbf{P4}, the shortest cycle $SC$ that goes through $\{ c_1, \dots, c_d \}$ has weight at most $2\sum_i \| C_i \| = 2 \| C_m \|$.
	Moreover, $SC$ is not longer than the shortest Hamiltonian cycle on $\mathcal{V}$, which in turn is not longer than $SHP$. Therefore,
	\[ \w(BC) - 2\|C_m\| \leq \sum_i |c_i| - 2\|C_m\| \leq \|SC\| \leq \| SHP \|. \]
\end{proof}

Combining~\eqref{eq:main1} and~\eqref{eq:main2}, we finally obtain
\[ \|\mathrm{PATH}(\mathcal{G})\| \leq 2\cdot\|SHP\| + \w(BC) - \| C_m \| \leq 3\cdot\|SHP\| + \| C_m \| \leq 4\cdot\|SHP\|.\]

\bibliographystyle{splncs04}
\bibliography{main}
	
\end{document}